\documentclass[reqno,11pt]{amsart}

\usepackage{amssymb,graphics,graphicx,wrapfig}
\usepackage{color}
\usepackage{soul}
\usepackage[nice]{nicefrac}

\def\eps{\varepsilon}

\def\be{\begin{equation}}
\def\ee{\end{equation}}
\def\ba{\begin{align}}
\def\bm{\begin{multline}}
\def\bfig{\begin{figure}[htb]}
\def\efig{\end{figure}}
\numberwithin{equation}{section}
\newtheorem{theorem}{Theorem}[section]
\newtheorem{proposition}[theorem]{Proposition}
\newtheorem{lemma}[theorem]{Lemma}

\DeclareMathSymbol{\leqslant}{\mathalpha}{AMSa}{"36}
\DeclareMathSymbol{\geqslant}{\mathalpha}{AMSa}{"3E}
\DeclareMathSymbol{\doteqdot}{\mathalpha}{AMSa}{"2B}
\DeclareMathSymbol{\circlearrowright}{\mathalpha}{AMSa}{"08}
\DeclareMathSymbol{\subsetneq}{\mathalpha}{AMSb}{"28}
\DeclareMathSymbol{\supsetneq}{\mathalpha}{AMSb}{"29}
\renewcommand{\leq}{\;\leqslant\;}
\renewcommand{\geq}{\;\geqslant\;}

\newcommand{\dd}{{\rm d}}
\newcommand{\e}[1]{\,{\rm e}^{#1}\,}

\newcommand{\upchi}{\raise 2pt \hbox{$\chi$}}

\def\writefig#1 #2 #3 {\rlap{\kern #1 truecm \raise #2 truecm
\hbox{#3}}}

\newcommand{\caM}{{\mathcal M}}

\newcommand{\bbR}{{\mathbb R}}

\newcommand{\bsB}{{\boldsymbol B}}
\newcommand{\bsC}{{\boldsymbol C}}

\begin{document}

%{\hfill\small \version} \vspace{2mm}

\title[Shape of emerging condensate in effective models of condensation]{The shape of the emerging condensate \\ in effective models of condensation}

\author{Volker Betz, Steffen Dereich and Peter M\"orters}% and Daniel Ueltschi}
%\address{Volker Betz \hfill\newline
%\indent FB Mathematik, TU Darmstadt \hfill\newline
%{\small\rm\indent http://www.mathematik.tu-darmstadt.de/$\sim$betz/} 
%}
%\email{betz@mathematik.tu-darmstadt.de}

%\address{Volker Betz \hfill\newline
%\indent Department of Mathematics \hfill\newline
%\indent University of Warwick \hfill\newline
%\indent Coventry, CV4 7AL, England \hfill\newline
%{\small\rm\indent http://www.maths.warwick.ac.uk/$\sim$betz/} 
%}
%\email{v.m.betz@warwick.ac.uk}

\maketitle

\begin{quote}
{\small
{\bf Abstract.}
We consider effective models of condensation where the condensation occurs 
as time~$t$ goes to infinity. We provide natural conditions under which the build-up of the 
condensate occurs on a spatial scale of $1/t$ and has the universal 
form of a Gamma density. The exponential parameter of this density is determined
only by the equation and the total mass of the condensate, while the power law 
parameter may in addition depend on the decay properties of the
initial  condition near the condensation point. We apply our results to some examples,
 including simple models of Bose-Einstein condensation. 
}  % end \small

\vspace{4mm}
\noindent
{\footnotesize {\it Keywords:} Emergence, kinetic equation, qunatum particles,  Bose-Einstein,
House-of-cards model, selection, mutation, singular solution, non-linear partial differential equation, 
non-equilibrium phenomena.

 }

\vspace{1mm}
\noindent
{\footnotesize {\it 2010 Math.\ Subj.\ Class.:} 35F25 (primary); 82C05; 82C10; 82C26; 35Q40;  82C40. }
\end{quote}

%{\footnotesize\tableofcontents}
\section{Motivation and background}

Condensation is an important and interesting phenomenon, which is present in many 
different physical systems. Loosely speaking, condensation occurs when for a 
system of many particles, a given relevant quantity has the same 
value for a macroscopic fraction of those particles. The examples that are 
relevant to this article are natural selection and mutation, 
where the particles are individuals, the relevant quantity is the fitness, 
and the condensation occurs at the maximal fitness; and Bose-Einstein 
condensation, where the particles are Bosons, the relevant quantity is the 
energy, and the condensation occurs at the minimal quantum energy level. 
We are interested in what a condensing system looks like when 
it is dynamically close to condensation. 

We will not deal with the difficult issue of condensation in many-particle 
systems itself, but will instead investigate effective models  and 
manifestations of condensation in  a scalar variable.  These models are 
then given in the form of  non-linear measure-valued equations in $1+1$ dimensions, 
the first of the variables being time, and the other the scalar quantity mentioned above. 
More precisely, for $t \in \bbR_0^+$, the finite measure $p_t$ describes the concentration 
of particles at relevant quantity $x$, and $(p_t)_{t\geq 0}$ solves the equation 
\begin{equation}
	\label{general evolution}
	\partial_t p_t(\dd x) = F(x,p_t).
\end{equation}
with some functional $F$. 
We will give three concrete examples of such systems below: Kingman's model
of selection and mutation \cite{K78}, an approximate model for Bose-
Einstein condensation due to Buffet, de Schmedt and Pul\'e \cite{BSP1} 
(henceforth called the BSP-model), and a model for bosons in a heat bath 
investigated by Escobedo, Mischler and Velazquez \cite{EM99,EM01,EMV04} 
(referred to als EMV-model below). 
Another effective model of condensation is the 
Boltzmann-Nordheim equation \cite{Lu04,Lu05,S10, EV15}, 
but as we will see below, it is too singular for our theory to apply. 

In the effective models, condensation is characterized by the behavior of 
$p_t$ as $t \to t^\ast$, where $t^\ast \leq \infty$ is the time at which 
condensation occurs. Typically, one assumes that the initial condition $p_0$ is absolutely continuous with respect to the Lebesgue measure, and 
$t^\ast$ is the infimum over all times where $p_t(\dd x) $ is not absolutely continuous with respect to the Lebesgue measure. 
If one finds that there exist $\rho > 0$ and $q \in L^1$ so that for suitable test functions $\phi$, 
\begin{equation}
	\label{condensation general}
	\lim_{t \uparrow t^\ast} \int p_t(\dd x) \phi(x) = \rho \phi(x^\ast) 
	+ \int q(x) \phi(x) \, \dd x, 
\end{equation}
this is paraphrased by saying that condensation occurs at $x = x^\ast$, at time $t = t^\ast$. 
We then refer to $\rho$ as the  \emph{mass of the condensate}, while 
$q$ is the \emph{bulk density}. The approach to the limit in 
Equation~\eqref{condensation general} can be interpreted as the formation of 
an approximate Dirac distribution around $x^\ast$ 
when $t$ is near $t^\ast$. The topic of our paper is to investigate 
in the case $t^\ast=\infty$ the asymptotic shape of this approximate Dirac distribution in the 
correct~scale. 

In cases where 
$t^\ast < \infty$, much less is known.  The 
best results that we are aware of are those by Escobedo and 
Velazquez for the Boltzmann-Nordheim equation \cite{EV15}. They 
show that under suitable assumptions on the mass of the 
initial condition, the solution $p_t(\dd x)$ to 
\eqref{general evolution} has a Lebesgue density $\tilde p_t$ that 
explodes in finite time, i.e.\ there exists a $t_\ast^- < \infty$ 
so that 
\[
\liminf_{t\uparrow t_\ast^-} \| \tilde p_t \|_\infty = \infty.
\]
They also show that there exists a $t_\ast^+ \geq t_\ast^-$ so that 
the weak solution of \eqref{general evolution} contains a Dirac mass 
at the condensation point $x^\ast = 0$ for some $t \in 
[t_\ast^-, t_\ast^+]$. What they cannot show is that the infimum over 
all possible $t_\ast^+$ equals $t_\ast^-$, and so an equation like 
\eqref{condensation general} is currently not known. 
We are not aware of 
any  natural examples where \eqref{condensation general} 
has been proved for some $t^\ast < \infty$ , and therefore stick
to the case $t^\ast=\infty$.

To our knowledge, the following rigorous results about the asymptotic shape of such
emerging condensates exist: In 
\cite{DeMo13}, it is shown that for Kingman's model of selection and 
mutation, the shape is the one of a Gamma-distribution. The same result is found in 
\cite{EMV04} for a special case of the EMV-model. They also find the influence of 
the initial condition, depending on its behavior near the condensation point, that will
appear in our results below. The results 
rely on explicit solution formulas for all times, although in \cite{EMV04} a formal 
asymptotic expansion is used in order to cover also cases where such a solution formula 
is missing. 

The contribution of our paper is to provide conditions which are easy to check,  
do not require the knowledge of a solution formula, and are sufficient to conclude 
both the Gamma shape of the near condensate and the possible 
dependence of that shape on the initial condition. Our conditions are natural in that 
they only require knowledge about the right hand side of \eqref{condensation general}
when $p_t$ is very close to the formal stationary solution $\rho \delta_{x^\ast} 
+ q$. Precise statements follow in the next section. \\[2mm]
{\bf Acknowledgements:} The authors would like to thank Daniel\ Ueltschi for many fruitful discussions and J.J.L.\ Velazquez for 
useful comments on the Boltzmann-Nordheim equation, and for pointing out reference 
\cite{EMV04}.

\section{Assumptions and main result}\label{sec2}

We treat models where condensation  occurs at the boundary of the set of possible values 
of the scalar quantity~$x$, and  
we normalize these models so that this value is $x=0$. Further we restrict attention to the case where $t^\ast=\infty$.
Let $\caM_0 := \{ \rho \,\delta_0 + p\, \dd x: \rho \geq 0, p \in L^1(\bbR^+_0)\}$ 
be the subspace of the space of finite measures on $\bbR^+_0$ that have 
a density with respect to Lebesgue measure except possibly at the origin where 
we allow for a Dirac measure~$\delta_0$ with mass $\rho$. Fix $\alpha > 0$, 
and let $\bsB: \caM_0 \to C(\bbR^+_0)$ and $\bsC: \caM_0 \to C(\bbR_0^+)$ be 
(not necessarily linear) operators. Consider the equation 
\be \label{the main equation}
\partial_t p_t(\dd x) = \bsB[p_t] p_t(\dd x) + x^\alpha \bsC[p_t] \, \dd x.
\ee
An element $p\in\caM_0$ is called \emph{stationary}, if
\[
\bsB[p] \,p(\dd x) + x^\alpha \bsC[p]\,\dd x=0.
\]
In most cases we consider evolutions where $p_t$ does not have an atom in $0$. In this  case we refer to the Lebesgue density of $p_t$ by the same symbol $p_t$. With this convention the equation reads 
\[
\partial_t p_t(x) = \bsB[p_t] p_t(x) + x^\alpha \bsC[p_t].
\]
When comparing \eqref{the main equation} to the most general 
equation \eqref{condensation general}, we see that we demand a decomposition 
of the right hand side into a homogenous part and a remainder. While such a
decomposition can always be achieved (e.g. by setting $\bsB = 0$), the 
restriction lies in the assumed regularity of the images under $\bsB$ and 
$\bsC$. In particular, we assume that when dividing the inhomogenous part by a 
factor of $x^\alpha$, we still retain a function that is bounded at $x=0$. 
In all the concrete and relevant 
examples that we are aware of, the decomposition is unique and easy to find. 
In the Boltzmann-Nordheim model, however, measures with a Dirac mass at the 
origin are too singular for the equation to make sense in a classical way 
\cite{EV15}, and so a decomposition like \eqref{the main equation}  with the
corresponding regularity assumptions fails. \\[2mm]
{\bf Definition:} We say that a  solution $(p_t)_{t\geq 0}$ to \eqref{the main equation} {\em converges regularly} to an element $p_\infty \in \caM_0$ if\\[1mm]
(i): $p_t \to p_\infty$ weakly as $t \to \infty$ as measures; \\[1mm]
(ii): the following two equations hold:
%For each $\delta > 0$, $\lim_{t \to \infty} 
%\int_{\delta}^\infty |p_t(\dd x) - p_\infty(\dd x)| = 0$;\\[2mm]
%We say that the convergence is {\em strongly regular} if 
%there exists $\delta > 0$ such that 
\begin{eqnarray}
	&& \lim_{t \to \infty} \| \bsB[p_t] - \bsB[p_\infty] \|_{C^1([0,\delta])} = 0, \label{A2_1} \\
	&& \lim_{t \to \infty} \| \bsC[p_t] - \bsC[p_\infty ] \|_{C([0,\delta])} = 0. \label{A2_2}
\end{eqnarray}
Here, $\| f \|_{C^1([0,\delta])} = \sup \{ |f(x)| + |f'(x)|: 0 \leq x \leq 
\delta\}$, and $\| f \|_{C([0,\delta])} = \sup \{ |f(x)|: 0 \leq x \leq \delta\}
$. 
%{\blue If $p_0$ has no atom at zero but $p_\infty$ has, the definitions of $\bsB[p_\infty]$ and  $\bsC[p_\infty]$ do not affect the dynamics of $(p_t)$. In this case equations (\ref{A2_1}) and (\ref{A2_2}) simply mean that a limit exists in the stated sense.}
%{\green I do not understand the last sentence. The limit should be equal to $p_\infty$, I think so just the existence is not enough.
%What are we trying to say?} 
\\%[2mm]

When $p_\infty$ is a measure with positive condensate mass, 
(i) above is what is usually proved when condensation is shown, 
see e.g. \cite{EM01,BSP1}. (ii) is more particular to our needs. 
Note that \eqref{A2_1} demands that 
the difference of $\bsB[p_t]$ and $\bsB[p_\infty]$ is of higher regularity 
than each individual term needs to be.

In two of the examples that we will give, the 
operators $\bsB$ and $\bsC$ are affine  
integral operators. In both of those 
examples, it has been shown that the convergence of $p_t$ to $p_\infty$ is 
in $L^1$ away from $x=0$. The next proposition states that in such cases, 
\eqref{A2_1} and \eqref{A2_2} already follow from natural 
regularity assumptions on the integral kernels of $\bsB$ and $\bsC$. 

\begin{proposition} \label{regularity from L1}
Assume that $\bsB$ and $\bsC$ are affine integral operators, i.e.\ 
\[
\bsB[p](x) = \int_0^\infty K_b(x,y) p(\dd y) + R_b(x), \quad \bsC[p](x) = 
\int_0^\infty K_c(x,y) p(\dd y) + R_c(x),
\]
where $K_b$, $K_c$, and $\partial_x K_b$ are elements of 
$C_b([0,\delta] \times \bbR_0^+)$ for some 
$\delta > 0$, while $R_b$ and $R_c$ are arbitrary functions.\\  
Let $(p_t)$ be a solution to \eqref{the main equation} 
with initial condition $p_0(x) \dd x$, i.e. without an atom at zero. 
Assume that $p_t$ converges weakly to $p_\infty = \rho \delta_0 + q(x) \dd x$ 
with $\rho \geq 0$ and $q \in L^1$, and in addition assume that 
\begin{equation} \label{L1 away from zero}  
  \lim_{t\to\infty} \int_\delta^\infty |p_t(x) - q(x)| \, \dd x = 0
\end{equation}
for all $\delta > 0$. Finally, assume that there exists $\tilde \delta > 0$
such that 
\begin{equation} \label{no stupid things near zero}
 	\limsup_{t \to \infty} \int_0^{\tilde \delta} |p_t(x)| \, \dd x < \infty.
\end{equation}
Then $p_t$ converges regularly to $p_\infty$. 	
\end{proposition}

The proof consists of standard applications of integral convergence theorems. 
We give it in the appendix for the convenience of the reader. Note that since 
$R_b$ and $R_c$ drop out when considering differences like $\bsB[p_t](x) - 
\bsB[p_\infty](x)$, they are indeed arbitrary, although the 
regularity requirements preceding \eqref{the main equation} will usually mean 
that they need to be continuous and bounded. 

Now we state our 
main result. 

\begin{theorem} \label{main thm}
	Assume that $(p_t)$ solves equation \eqref{the main equation} with initial
	condition $p_0(x) \, \dd x$, i.e.\ without atom at zero. 
	For the density $p_0$, 
	assume that there exists $\alpha_0 > 0$ and a function $\eta\colon \bbR^+_0 \to \bbR^+_0$, 
	which is continuous and positive at zero, so that 
	\[
	p_0(x) := x^{\alpha_0} \eta(x).
	\]
	Assume further 
	that $(p_t)$ converges regularly to a stationary limit 
	$p_\infty = \rho \delta_0 + q(x) \,\dd x$ with $\rho > 0$ and $q \in L^1$. 
	Finally, assume 
	that for this limit,
	\begin{equation}\label{power law at zero}
		c_1 := \bsC[p_\infty](0) > 0, 
	\end{equation}
	and that, for some $\alpha>0$,
	\begin{equation} \label{gamma}
		c_2 := \lim_{x \to 0} \frac{x^{\alpha-1}}
		{q(x)} > 0  \text{ exists.}
	\end{equation}
	Then with $\gamma := c_1 c_2$ and $\beta:= \min \{\alpha, \alpha_0\}$ we have, 
	uniformly on compact intervals of $\bbR_0^+$, that
	\begin{equation} \label{solution}
	\lim_{t \to \infty} \tfrac{1}{t} p_t(\tfrac{x}{t}) = C \e{- \gamma x} 
	x^{\beta}.
	\end{equation}
	Above, $C = \rho \gamma^{\beta} 
	/ \Gamma(\beta)$, i.e.\ $C$ 
	is such that the right hand side 
	of \eqref{solution} integrates to 
	$\rho$.
\end{theorem}

The main feature of this result is the universal nature of the gamma shape of the emerging 
condensate. In all examples that we are aware of, the stationary limit $p_\infty$ in Theorem \ref{main thm} only depends on the 
mass of the initial condition, but not on its shape. 
In these cases, the statement can also be read as the following dichotomy: either
$\alpha \leq \alpha_0$, in which case the initial condition is irrelevant for 
the shape of the emerging condensate; or, $\alpha_0 < \alpha$, in which case the 
exponential decay of the Gamma distribution is still governed by the constant 
from \eqref{gamma}, 
but the power law near $x=0$ is the same as in the initial condition. 
From the calculations in our proofs, it is apparent that the second case could 
be strengthened in the following way: if the initial condition dominates the 
inhomogeniety near $x=0$, the emerging condensate will look like the 
initial condition at an appropriate scale. In particular, one might think 
of initial conditions that switch between different power laws $\alpha_1$ and 
$\alpha_2$ infinitely often in the approach to zero. In that case, we would have no 
convergence in \eqref{solution}, but rather an oscillating behavior. Since this case
does not seem very relevant and would need rather careful statements and 
investigations, we do not pursue it any further. \pagebreak[3]

Note that condition \eqref{gamma} can alternatively be read as a condition on 
$\bsB[p_\infty](x)$ near the point~$x=0$. Namely, by stationarity of $p_\infty$ we have that %and continuity of $\bsB[p_\infty]$ and  $, we have 
\[
\bsB[p_\infty](x) = - \frac{x^\alpha}{q(x)} \bsC[p_\infty](x)
\]
for Lebesgue almost all $x\geq 0$ and by continuity this holds for \emph{all} $x\geq 0$. For $x=0$ the right hand side equals zero so that %Since in addition, $x^\alpha C[p_\infty](x)$ cannot cancel an atom 
%at $x=0$. we must have 
$\bsB[p_\infty](0) = 0$. Thus, condition \eqref{gamma} means
that $x \mapsto \bsB[p_\infty](x)$ is differentiable at $x=0$, and that its derivative
equals $-\gamma$. 

\section{Examples}

\subsection{Selection mutation equations}
 A natural set of examples for our theory arise from equations describing the fitness distribution of a 
population evolving by selection and mutation. We focus on Kingman's model~\cite{K78} and briefly mention some generalizations and variants at the end of this section. 

Kingman's model of selection and mutation is originally framed 
in discrete time, see~\cite{K78}. We start with an initial fitness distribution 
$p_0(\dd x)$  of a diffuse population, which is a probability measure on $(0,1)$.  
By $p_n(\dd x)$  we denote the fitness distribution in  the $n$-th generation. 
It satisfies the recursion 
\[
p_{n+1}(\dd x) = (1 - \beta) \frac{x}{w_n} p_n(\dd x) + \beta r (\dd x),
\]
where $w_n = \int_0^1 x p_n(\dd x)$ is the mean fitness at generation $n$, 
$r$ is the fitness distribution for spontaneous mutations, and $\beta \in (0,1)$ 
is the frequency of mutation. We assume that $r$ is a probability measure on $(0,1)$ 
with essential supremum at $x=1$. Kingman's idea is that a proportion $1-\beta$ of the population is selected from the previous generation with a selective advantage proportional to their fitness, and a proportion~$\beta$ of the population experiences mutation, which destroys the individuals' biochemical 
`house of cards' so that the mutant fitness distrbution~$r$ does not depend on their 
previous fitness. Kingman showed that condensation 
at the maximal fitness $x=1$ occurs if 
%\[
%\lim_{n \to \infty} \int x^n p_0(\dd x) / \int x^n r( \dd x) = 0,
%\]
%and 
\begin{equation}\label{condcrit}
%\gamma(\beta) := 1 - \beta \int_0^1 \frac{r(\dd x)}{1-x} > 0.
\beta \int_0^1 \frac{r(\dd x)}{1-x} <1.
\end{equation}
%In this case, $(p_n)$ converges to the limiting distribution 
%\[
%p_\infty(\dd x) = \gamma(\beta) \delta_1(\dd x) + \beta \frac{r(\dd x)}{1-x}.
%\]
To adapt the model to our framework, we switch to a continuous time model
and change variables so that condensation occurs at $x=0$. The result is the 
equation 
\be \label{kingman continuous}
 \partial_t p_t(\dd x) =\Big(  (1 - \beta) \frac{1-x}{w[p_t]} - 1\Big) p_t(\dd x) + 
\beta u(\dd x),
\ee
where $w[p_t] = \int_0^1 (1-x) p_t( \dd x)$ and
 $u(\dd x)=r(\dd(1-x))$.
 Kingman's arguments show that~if $$\beta \int_0^1 \frac{u(\dd x)}{x} < 1$$ we have
$p_t \to p_\infty$ weakly for the  
stationary solution $p_t$ %to equation \eqref{kingman continuous} 
given by
\[
p_\infty(\dd x) =  \Big(1 - \beta \int_0^1 
\frac{u(\dd x)}{x}\Big) \delta_0 + \beta \frac{u(\dd x)}{x}.
\]
%
%As in our framework, the only 
%place where $u$ can be singular to the Lebesgue measure is $x=0$, but on 
%the other hand the integral $\int_0^1 \frac{u(\dd x)}{x}$ needs to be finite,
%we see that the spontaneous mutation distribution must be absolutely 
%continuous with respect to Lebesgue, and must converge to zero 
%quickly enough as $x$ converges to the maximal inverse fitness $x=0$. 
%
We assume that $u(x) = x^\alpha u_0(x)$, for some $\alpha > 0$,
and $u_0(x)$ continuous and strictly positive  at $x=0$. 
Then, 
\[
\bsB[p] = (1 - \beta) \frac{1-x}{w[p]} -1,
\] 
and $\bsC[p] = \beta u_0$. In particular \eqref{A2_1}
and \eqref{A2_2} are trivially fulfilled. Moreover,
\[
w[p_\infty] = \Big( 1 - \beta \int_0^1 \frac{u(\dd x)}{x}\Big) + \beta \int \frac{1-x}{x} u(\dd x) 
= 1-\beta,
\]
as $u(\dd x)$ is a probability measure.  Then, 
$\bsB[p_\infty](x) = - x$, and 
$\bsC[p_\infty](x) = \beta u_0(x)$, confirming
that $p_\infty$ is stationary. Moreover, the bulk density $q(x)$ is 
given by $\beta u(x) / x = \beta x^{\alpha - 1} u_0(x)$. Thus, $c_2$ from 
\eqref{gamma} is given by $1/\beta u_0(0)$, and 
Theorem~\ref{main thm} holds with 
$\gamma = 1$. We thus obtain the gamma-shape of the emerging condensate under  less restrictive conditions than Dereich and M\"orters~\cite{DeMo13}.
\smallskip%
%{\red\sf Can we easily (re-)prove (ii) in the definition of regular convergence? 
%Can we see what happens when the initial condition dominates at $x=0$? More
%precisely, do we still have regular condensation in this case, so that our 
%theory applies?} 

Kingman's model was generalized by  Yuan~\cite{Y17} to model the 
long-term evolution of Eschrichia.coli in the Lenski experiment. This 
model can still be fitted to our framework. Park and Krug~\cite{PK08} 
generalize Kingman's model to unbounded fitness distributions, which 
leads to a qualitatively different emergence of a condensate (at 
infinity).  They also investigate the relation of Kingman's model to a stochastic finite 
population model.  Dereich~\cite{Der16} identifies the shape of the emergent condensate 
in a random network model with fitness.
\smallskip

%\cite{PK08} 
We now discuss in more detail a stochastic particle model based on a branching process closely related to Kingman's original model, which is investigated in \cite{DMM17}. In this model immortal particles produce offspring with a rate given by their fitness. Independently, each offspring particle is a mutant with probability~$\beta$ and otherwise inherits the parent's fitness. Mutants receive their fitness by sampling from the distribution~$r$. The expected fitness distribution~$q_t$ of particles alive at time~$t$ therefore satisfies
$$\dot{q_t}(\dd x)= (1-\beta) x q_t(\dd x) + \beta \int y \, q_t(\dd y) r(\dd x).$$ 
Taking $a_t=\int_0^t \int y \, q_s(\dd y)\, \dd s$ the normalized and time-changed quantity 
$$\tilde{p}_t=\frac{q_{a_t^{-1}}}{\int q_{a_t^{-1}}(\dd y)}$$ 
therefore satisfies
$$\dot{\tilde{p}}_t(\dd x)= \Big( (1-\beta)  \frac{x }{\int y \, \tilde{p}_t(\dd y)} 
-1 \Big) {\tilde p}_t(\dd x) + \beta r(\dd x).$$ 
Moving now the condensation point to the origin we are back in Kingman's model 
and Theorem~2.1 can be applied. If~\eqref{condcrit} holds we have $a_t\sim (1-\beta) t$ 
and we observe that also in the stochastic model the condensate is forming on the 
scale $1/t$ in expectation.  The behaviour of this stochastic  model in probability is
more difficult to  investigate and largely open.

%Mention Hofbauer,  J. Math. Biology (1985) 23:41-53.
%Also Crow and Kimura. And B\"urger.  And expectations of the stochastic model. 

\subsection{A model of Bosons in contact with a bath of Fermions}

In a series of papers \cite{EM99,EM01,EMV04}, Escobedo, Mischler and Velazquez 
study an effective model for Bosons in contact with a bath of Fermions in 
thermal equilibrium. Given a  function $b:\bbR_+^0\times \bbR_+^0\to \bbR_+$ of the form
$$
b(x,y)= \e{\eta x}\e{\eta y} \sigma(x,y)
$$
with $\eta\in[0,1)$ and $0\leq \sigma \in L^\infty(\bbR_+^0 \times \bbR_+^0)$ symmetric, one asks for solutions $(F_t)_{t\geq 0}$ of the equation
$$
\partial_t F_t(x)=\int_0^\infty b(x,y)\Big(F_t(y)(x^2+F_t(x))e^{-x}-F_t(x)(y^2+F_t(y))e^{-y}\Big)\,\dd y.
$$
Note that the right hand side is well defined as long as $F_t$ is in the space $L^1(\bbR_0^+, \e{\eta x} \,\dd x)$.
To apply our results we consider the transformed solution $(p_t)_{t\geq 0}$ given by $p_t(x)=\e{\eta x}F_t(x)$. It satisfies as $L^1(\bbR_+^0)$-valued system the equation
\begin{equation}\label{EMV-equation}
	\partial_t p_t(x) = \e{\eta x} \int_0^\infty \sigma(x,y)
	\Big(  p_t (y) (x^2 \e{\eta x} + p_t(x)) \e{-x} - p_t(x) \big(y^2 \e{\eta y} + p_t(y)\big) 
	\e{-y} \Big) \, \dd y.
\end{equation}
%with $$
%\sigma(x,y)=\e{-\eta x}\e{-\eta y} b(x,y)
%$$
%
%By coceiving $F_t$ as measure and writing By introduction of a density we an application of 
%
%The equation, written for the case when $(F_t)_{t\geq 0}$ is an element
%of $L^1(\bbR_0^+)$ (identified with a measure by declaring $p$ to be the density),
%reads 

The extension of this equation to elements of $\caM_0$ is straightforward, 
and it fits into the framework of \eqref{the main equation} with $\alpha = 2$,
\be \label{EMV B}
\bsB[p](x) = \e{\eta x}\int_0^\infty  p(\dd y) \sigma(x,y) \big( \e{-x} - \e{-y} \big) 
- \e{\eta x}\int_0^\infty \sigma(x,y) y^2 \e{(\eta-1)y} \, \dd y,
\ee
and 
\begin{equation} \label{EMV C}
\bsC[p](x) = \e{(2\eta-1) x} \int_0^\infty \sigma(x,y) p(\dd y).
\ee
Under suitable assumptions on $\sigma$  and the initial condition, 
Escobedo and Mischler \cite{EM99, EM01}
show existence of solutions for all times, and convergence of the solution 
as $t \to \infty$ towards the stationary solution. In the case where $\int_0^\infty F_0(x)\,\dd x>\int_0^\infty \frac{x^2}{\e{x} - 1}\,\dd x =2.40411...$ condensation occurs and the limit is of the form
$$
F_\infty(\dd x)=\rho \delta_0 + Q(x) \dd x \quad \text{with } Q(x) = 
\frac{x^2}{\e{x} - 1}
$$ 
and $\rho=\int_0^\infty F_0(x)\,\dd x-\int \frac{x^2}{\e{x} - 1}\,\dd x$. 
%{\sf (Please check formulas.)} 
In that case $(p_t)$ has as limit
\begin{align}\label{limit_EMV}
p_\infty(\dd x) = \rho \delta_0 + q(x) \dd x \quad \text{with } q(x) = 
\frac{x^2 \e{\eta x}}{\e{x} - 1}.
\end{align}
In \cite{EMV04}, the authors investigate the shape of the emerging condensate; 
they show rigorously (by means of an explicit solution formula) that the 
emerging condensate is Gamma-shaped in the special case that $b(x,y) \equiv 1$.
Moreover, they find that the power law parameter of the Gamma distribution 
depends on how the initial condition vanishes at $x=0$; this corresponds to Theorem 
\ref{main thm} of the present paper. Very interestingly, they also 
obtain asymptotics for the case where the initial condition already has a 
Dirac mass at $x=0$. In that case, there is no emerging condensate as 
$t \to \infty$, instead the Dirac mass itself grows to its final value 
(determined by the total mass of the initial condition) as $t \to \infty$; 
see Theorem 1, part (ii) of \cite{EMV04}. This result shows that we cannot 
drop the assumption on absolute continuity of the initial condition in 
Theorem \ref{main thm}. 

We will now use our general theory in order to find more general conditions 
on $b$ so that the Gamma shape of the emerging condensate still holds. 
We first assume that $(x,y) \mapsto \sigma(x,y)$ is continuous, 
continuously differentiable with respect to $x$ for all $y$, and that,
 for some $\delta > 0$, 
\begin{equation}\label{a}
\sup \{ | \partial_x \sigma(x,y)|: 0 \leq x \leq \delta, y \in \bbR^+\} < 
\infty.\end{equation}
This condition is tailor-made for the assumptions on $\bsB$ and $\bsC$ 
of Proposition ~\ref{regularity from L1}. In fact it is slightly more
than we really need, since we have 
\[ 
K_b(x,y) = \sigma(x,y) (\e{-x} - \e{-y}),
\]
and thus at 
$x=y$ differentiability of $\sigma$ is not needed due to the presence of the 
factor $\e{-x} - \e{-y}$. We further assume that
\begin{equation}\label{b}
\int_0^\infty \sigma(0,y) p_\infty(\dd y)  > 0.
\end{equation}
This assumption  ensures the validity of conditions~\eqref{power law at zero} and~\eqref{gamma}.
Indeed, we compute
\[
\bsC[p_\infty](x) = \e{(2\eta-1)x} \Big( \rho \sigma(x,0) + \int_0^\infty \sigma(x,y) 
\frac{y^2 \e{\eta y}}{\e{y}-1} \, \dd y \Big), 
\]
and thus $\bsC[p_\infty](0)>0$, so~\eqref{power law at zero} holds. Since 
\[
\lim_{x\to0} \frac{x^{\alpha - 1}}{q(x)} =\lim_{x\to0} \e{\eta x}\frac{\e{x}-1}{x} =1, 
%\Big( \rho \sigma(x,0) + \int_0^\infty \sigma(x,y) 
%\frac{y^2\e{\eta y}}{\e{y}-1} \, \dd y \Big).
\]
we also have \eqref{gamma}, and find that 
$\gamma = \int \sigma(0,y) p_\infty(\dd y)$. Note that~\eqref{b} is weaker
than the assumption $\sigma(0,0) > 0$ that was made in the appendix of 
\cite{EMV04}, where non-constant $b$ are treated non-rigorously, 
using matched asymptotic 
expansions.

Finally we need to verify~\eqref{L1 away from zero}. This is proved by  Escobedo and Mischler in
Theorem~6 of \cite{EM01} for the case $\eta=0$ and we
believe that the arguments given there allow to prove 
\eqref{L1 away from zero} also for arbitrary~values of $\eta>0$.
In conclusion, if conditions \eqref{a} and \eqref{b} hold
and $\eta=0$, % remove if appropriate
then Theorem \ref{main thm}
holds with $\alpha = 2$ and $\gamma = \int \sigma(0,y) p_\infty(\dd y)$. %We 

\subsection{Kinetics of Bose-Einstein condensation}

Our final example is a simple model for the emergence of a condensate in a Bose 
gas in contact with a heat bath, which was developed by Buffet, 
de Smedt and Pul\'e in~\cite{BSP1, BSP2}. 

Let $\hat C:\bbR \to \bbR_0^+$ be a strictly positive, bounded function satisfying 
the \emph{KMS relation}
\begin{equation} \label{bath corr}
\hat C(-z) = \hat C(z) \e{\beta z}
\end{equation}
for some $\beta > 0$. Physically, $\hat C$ is the Fourier transform of the heat 
bath correlation function, and $\beta>0$ is the inverse temperature of the heat 
bath. Note also that the assumption that $\hat C$ is bounded and 
\eqref{bath corr} imply exponential decay of $\hat C(z)$ as $z \to \infty$. 

We also define 
\[
\begin{split}
&A(z) = \hat C(-z) - \hat C(z) = \hat C(z) (\e{\beta z} - 1), \\
& F(x)= \frac{x^{\nicefrac12} }{\sqrt{2}\,\pi^2}.
\end{split}
\]
$F$ is the density of states in the case of a Bose gas in a 3-dimensional box.
In order to model the Bose gas in other environments we would need to modify 
$F$; for example, for the 
Bose gas in a $3$-dimensional harmonic trap, $F$ would decay like 
$x^{3/2}$ near $x=0$; see also the discussion in~\cite{CD14}. In what 
follows we will consider the more general form 
\begin{equation} \label{BSP alpha}
F(x) = f(x) x^\alpha, \qquad \alpha > 0, \quad  f(0) > 0.
\end{equation}
The energy distribution $p_t$
of the Bose gas at time $t\geq 0$ then satisfies the  equation
\be \label{BSP equation}
\begin{aligned}
\partial_t p_t(x)= \int_0^\infty & A(y-x) p_t(x)  p_t(y) \, \dd y \\  & - \int_{0}^\infty \hat
C(y-x) F(y) \, p_t(x) \, \dd y   +  \int_0^\infty  \hat C(x-y) p_t(y) \, F(x) \,  \dd y.
\end{aligned}
\ee
It is easy to check that \eqref{BSP equation} preserves the total
mass $\int_0^\infty p_t(x) \, \dd x$ for all times~$t\geq 0$.
It has been shown in
\cite{BSP1} that there is a global solution to \eqref{BSP equation} and moreover, 
there exists $\rho_{\rm c}>0$ such that, for all initial energy densities satisfying
\[
\int_0^\infty p_0(x) \, \dd x > \rho_{\rm c} + \rho
\]
for some $\rho > 0$, $(p_t)$ converges weakly to the stationary solution 
\[
p_\infty(\dd x) = \rho \delta_0 + q(x) \, \dd x
\]
with bulk density 
\[
q(x) = \frac{F(x)}{\e{\beta x} - 1}.
\]
Moreover, Theorem 2 of \cite{BSP1} states that 
\[
\lim_{t\to\infty} \int_c^\infty |p_t(y) - q(y)| \, \dd y = 0
\]
for all $c>0$. Thus, \eqref{L1 away from zero} and \eqref{no stupid things near zero} hold. 

The decomposition of equation~\eqref{BSP equation} according to 
\eqref{the main equation} is given by 
\[
\begin{split}
&\bsB[p](x) = \int_0^\infty A(y-x) p(y) \, \dd y - \int_{0}^\infty \hat C(y-x) F(y) \, \dd y,\\
&\bsC[p](x) =  f(x) \int_0^\infty  \hat C(x-y) p(y) \, \dd y.
\end{split}
\]
So in the context of Proposition~\ref{regularity from L1}, we have 
$K_b(x,y) = A(y-x)$, $K_c(x,y) = \hat C(x-y)$, $R_1(x) = 
\int_{0}^\infty \hat C(y-x) F(y) \, \dd y$ and $R_2(x) = 0$.  

We now assume that the restriction of  $\hat C$ to the interval $(-\infty,0)$ is 
continuously differentiable with a bounded derivative, and that $f$ is 
continuous. Then the integral kernels $K_b(x,y) = A(y-x)$ and $K_c(x,y) = \hat C(x-y)$ 
are bounded 
and continuous. In particular, $\bsB$ and $\bsC$ map $\caM_0$ to $C(\bbR_0^+)$.
Moreover, $x \mapsto A(y-x) = \hat C(x-y) - \hat C(y-x)$ is continuously 
differentiable with bounded derivative whenever $x \neq y$. 
For the case $x=y$, note that 
$A(z) = \hat C(z) (\e{\beta z} - 1)$, and thus $x \mapsto A(x-y)$ is 
differentiable also at $x=y$ with derivative $\beta \hat C(0)$. Altogether, 
we get that $\partial_x K_b(x,y)$ is in $C_b([0,\delta] \times \bbR_0^+)$, and 
Proposition~\ref{regularity from L1} implies that $p_t$ converges regularly to~
$p_\infty$.  

It remains to check assumptions \eqref{power law at zero} and \eqref{gamma}.
Clearly, 
\[
\bsC[p_\infty](0) = f(0) 
\int \hat C(y) p_\infty(y) \dd y > 0,
\]
which shows that $\alpha$ is determined by the choice in \eqref{BSP alpha}. 
Furthermore, 
\[
\lim_{x \to 0} \frac{x^{\alpha-1}}{q(x)} = \lim_{x \to 0} 
\frac{\e{\beta x}-1}{x f(x) } = \frac{\beta}{f(0)},  
\]
which shows that 
\[
\gamma = \beta \int_0^\infty \hat C(y) p_\infty(y) \, \dd y.
\]
By Theorem \ref{main thm}, we conclude again that
the emerging condensate is described 
by a Gamma distribution on the relevant scale.

\section{Proof of the theorem}

Let $(p_t)$ be a solution to \eqref{the main equation}. 
We define, for $x \geq 0$,
\[
b_t(x) := \bsB[p_t](x),
\]
and, for $x > 0$,
\[
c_t(x) = \bsC[p_t](x).
\]
Then $p_t$ solves the time-inhomogenous linear equation 
\[
\partial_t p_t(x) = b_t(x) p_t(x) + x^\alpha c_t(x)
\]
with initial condition $p_0$, and thus has the representation 
\begin{equation} \label{inhomogenous equation}
p_t(x) = \int_0^t \e{\int_s^t b_u(x) \, \dd u} x^\alpha c_s(x) \, \dd s + 
\e{\int_0^t b_s(x) \, \dd s} p_0(x). 
\end{equation}
We will re-write this representation in a form that will be convenient later 
on, using the following definitions: 
\[
\begin{split} 
W_t & := \e{\int_0^t b_s(0) \, \dd s}, \\
\gamma_t(x) & := - \frac{1}{x} (b_t(x) - b_t(0)) \quad \text{for } x>0,\\
\bar \gamma_{s,t}(x) & := \begin{cases}
	\frac{1}{t-s} \int_{s}^t \gamma_r(x) \, \dd r & \text{if } t >s\\
	\gamma_t(x) & \text{if } t = s. \end{cases}
\end{split}
\]
Then we have 
\begin{equation} \label{rec}
p_t(x) =  \int_0^t  \frac{W_t}{W_s}   x^\alpha c_s(x)   
\e{- (t-s) x \bar \gamma_{s,t}(x)}\,\dd s + W_t \e{- t x \bar \gamma_{0,t}(x)} 
p_0(x), 
\end{equation}
Next, we use our assumptions in order to prove some properties of the 
quantities $c_t$ and~$\gamma_t$. 

\begin{proposition} \label{b,c from B,C}
	Let the assumptions of Theorem \ref{main thm} be fulfilled. 
	Then 
	\begin{enumerate}
		\item $\lim_{t \to \infty} b_t(0) = 0$.
		\item For sufficiently large 
		$t$, $x \mapsto \gamma_t(x)$ can be continuously 
		extended to all $x \geq 0$. Further, there exists $\delta > 0$ and a 
		strictly positive function $\gamma_\infty: [0,\delta] \to \bbR^+$ with 
		the property that 
		\[
		\lim_{t \to \infty} \sup_{x \in [0,\delta]} 
		|\gamma_t(x) - \gamma_\infty(x)| = 0.
		\] 
		Moreover, we have that $\gamma_\infty(0) = \gamma$, the latter being
		defined in \eqref{gamma}.
		\item There exists $\delta > 0$ and a continuous function 
		$c_\infty: [0,\delta] \to \bbR$ such that
		$c_\infty(0) > 0$, and
		\[
		\lim_{t \to \infty} \sup_{x \in [0,\delta]} |c_t(x) - c_\infty(x)| = 0. 
		\]  
	\end{enumerate}	
\end{proposition}
\begin{proof}
	By the comments at the end of Section~\ref{sec2} we have $\bsB[p_\infty](0)=0$
%(1) follows from the comments at the end of Section~\ref{sec2}.}% As $p_\infty$ is a stationary solution it follows that }
%Since the right hand side of \eqref{the main equation} is zero when 
%	$p_\infty$ is inserted {\red\sf ($\bsB[p_\infty]$ and $\bsC[p_\infty]$ only enter via (2.2) and (2.3) and it is not clear that $p_\infty$ defines a stationary solution)}  , and since $x^\alpha \bsC[p_\infty]$ 
%	is a continuous function
%	and in particular has no atom at $x=0$, we conclude that 
%	$\bsB[p_\infty](0) = 0$. S
so that  by \eqref{A2_1}%, we find that 
	\[
	\lim_{t \to \infty} b_t(0) = \lim_{t\to\infty} \bsB[p_t](0) = 
	\bsB[p_\infty](0) = 0,
	\]
	%{\red\sf (To make this argument rigorous one would need continuity of $\bsB$ in an appropriate sense.)}
	and (1) is shown. For (2), we define 
	\[
		\gamma_{\infty}(x) := - \frac{1}{x} \bsB[p_\infty](x) \quad 
		\text{for } x > 0. 		
	\]
	Since $\bsB[p_\infty](x) q(x) = - x^\alpha 
	\bsC[p_\infty](x)$ for all $x > 0$, we 
	find that  
	\[
	\lim_{x \to 0} \gamma_\infty(x) = \lim_{x \to 0} \frac{x^\alpha}{x q(x)} 
	\bsC[p_\infty](x) = \gamma
	\]
	by Assumption \eqref{gamma} and the definition of $\gamma$. In particular,
	$\gamma_\infty$ can be continuously 
	extended to the whole half line. We write $b_\infty$ for $\bsB[p_\infty]$,
	and put 
	\[
	s_t(x) := b_\infty(x) - b_t(x).
	\]
	Then by \eqref{A2_1}, 
	$x \mapsto s_t(x)$ is an element of $C^1[0,\delta]$ for some 
	$\delta > 0$ and all sufficiently large $t$, and 
	$\lim_{t \to \infty} \|s_t \|_{C^1([0,\delta])} = 0$. 
	Furthermore, 
	\[
	\gamma_t(x) - \gamma_\infty(x) = \frac{1}{x} (s_t(x) - s_t(0)) = 
	\int_0^1 s_t' (r x) \, \dd r,
	\]
	and thus (2) follows. (3) follows directly from Assumption 
	\eqref{A2_2}.
\end{proof}

Recall our assumption on the initial condition $p_0$: 
there exists $\alpha_0>0$
and $\eta \in L^1$ so that $\eta$ is continuous at $x=0$, $\eta(0) > 0$, and 
\begin{equation}
\label{initial condition}	
p_0(x) = x^{\alpha_0} \eta(x)
\end{equation}
for all $x \in [0,\delta]$. Let us also define, for $\beta > 0$, 
\begin{equation}
\label{Q_t}
Q_t(\beta) := W_t^{-1} (t+1)^{1+\beta}.
\end{equation}
A direct calculation then gives an expression for 
$p_t/t$ at arguments of the order $1/t$, namely
\begin{equation}
\label{Q formula}
\tfrac{Q_t(\beta)}{t} p_t(\tfrac{x}{t}) = \big(\tfrac{t+1}{t}\big)^{1+\beta} 
\e{- x \gamma_\infty(0)} \Big( t^{\beta - \alpha} x^\alpha 
J(x,t)
+ t^{\beta - \alpha_0} x^{\alpha_0} 
\e{x(\gamma_\infty(0) - \bar \gamma_{0,t}(x/t))} \eta(\tfrac{x}{t}) \Big),
\end{equation}
with 
\[
J(x,t) = \int_0^t \frac{Q_s(\beta)}{(s+1)^{\beta+1}} c_s(x/t) 
\e{x (\gamma_\infty(0) - \frac{t-s}{t} \bar \gamma_{s,t}(x/t))} \, \dd s.
\]

Now it is easy to prove the following 

\begin{proposition}
\label{cond shape}
	Assume that the quantities $(b_t)$ and $(c_t)$ have the properties 
	(1) - (3) from Proposition \ref{b,c from B,C}, and assume in addition that for 
	$\beta := \min \{ \alpha, \alpha_0 \}$, the limit 
	$Q_\infty = \lim_{t \to \infty} Q_t(\beta)$ exists in $(0, \infty)$. 
	Then
	\begin{equation}
	\label{condensate shape}
	\lim_{t \to \infty} \frac{1}{t} p_t(\tfrac{x}{t}) 
	= \frac{1}{Q_\infty} \e{-x\gamma_\infty(0)  } 
	\Big( 1_{\{\beta = \alpha\}} x^\alpha \int_0^\infty 
	\frac{Q_s(\beta)}{(s+1)^\beta} c_s(0) \, \dd s +  1_{\{\beta = \alpha_0\}}
	x^{\alpha_0} \eta(0) \Big),		
	\end{equation}
	and the limit is uniform in $x$ on compact intervals in $(0,\infty)$.
\end{proposition}

\begin{proof}
First we analyse the first summand on the right hand side of~(\ref{Q formula}). 
We show that 
	$J(x,t)$ converges uniformly to 
	$$ \int_0^\infty \frac{Q_s(\beta) }{(s+1)^\beta} c_s(0)\,\dd s.
	$$ 
	To see this, note that
	by Proposition~\ref{b,c from B,C}, $\bar \gamma_{0,t}$ converges uniformly on $[0,\delta]$ to $\gamma_\infty$ and by continuity of $\gamma_\infty$ in zero we get local uniform convergence of $\bar\gamma_{0,t}(x/t)$ to $\gamma_\infty(0)$. Hence  for each fixed $s$,  the quantity  
	\[
	D(s,t) :=  \sup_{0 \leq x \leq K} \Big| c_s(x/t) 
	\e{x (\gamma_\infty(0) - \frac{t-s}{t} \bar \gamma_{s,t}(x/t))} - c_\infty(0)
	\Big|
	\]
	converges to zero as $t \to \infty$. The integrand in $J(x,t)$ and 
	$\frac{Q_s(\beta)}{(s+1)^\beta} c_s(0)$ are both bounded by a 
	constant multiple of the integrable function $(s+1)^{-1-\beta}$, uniformly in $x$, 
	on compact intervals. Thus for $M > 0$,  
	\[ \sup_{0\leq x \leq K}
	\Big| J(x,t) - \int_0^\infty \frac{Q_s(\beta)}{(s+1)^\beta} c_s(0) \, \dd s 
	\Big| \leq \int_0^M \frac{Q_s(\beta)}{(s+1)^\beta} D(s,t) \, \dd s + 
	C \int_M^\infty (s+1)^{-1-\beta} \, \dd s.
	\]
		The first term above converges to zero by dominated convergence, as 
	$t \to \infty$. By taking $M \to \infty$ afterwards, the asymptotic of the first term follows. 
%By Proposition~\ref{b,c from B,C}, $\bar \gamma_{0,t}$ converges uniformly on $[0,\delta]$ to $\gamma_\infty$ and by continuity of $\gamma_\infty$ in zero we get local uniform convergence of $\bar\gamma_{0,t}(x/t)$ to $\gamma_\infty(0)$. Moreover, $\eta$ is continuous in zero so 
Using  that, uniformly on compact intervals,
$
\e{x(\gamma_\infty(0) - \bar \gamma_{0,t}(x/t))} \eta(\tfrac{x}{t}) \to \eta(0)
$
we get convergence of the second term.
\pagebreak[3]

%
%	Inspecting \eqref{Q formula}, we see that the indicator functions arise from 
%	the terms $t^{\beta - \alpha}$ and $t^{\beta - \alpha_0}$, respectively. 
%	Since $\eta$ is continuous at $x=0$ by assumption, we have 
%	$\sup_{0 \leq x \leq K} | \eta(x/t) - \eta(0) |  \to  0$ as $t \to \infty$,
%	otherwise we could find a sequence $(x_n)$ converging to $0$ so that 
%	$\eta(x_n)$ does not converge to $\eta(0)$. It remains to show that 
%	$J(x,t)$ converges uniformly to the appropriate value. To see this, first 
%	note that for each fixed $s$,  the quantity  
%	\[
%	D(s,t) :=  \sup_{0 \leq x \leq K} \Big| c_s(x/t) 
%	\e{x (\gamma_\infty(0) - \frac{t-s}{t} \bar \gamma_{s,t}(x/t))} - c_\infty(0)
%	\Big|
%	\]
%	converges to zero as $t \to \infty$. Moreover, the integrand in $J(x,t)$ and 
%	$\frac{Q_s(\beta)}{(s+1)^\beta} c_s(0)$ are both bounded by a 
%	constant multiple of the integrable function $(s+1)^{-1-\beta}$, uniformly in $x$, 
%	on compact intervals. Thus for $M > 0$,  
%	\[ \sup_{0\leq x \leq K}
%	\Big| J(x,t) - \int_0^\infty \frac{Q_s(\beta)}{(s+1)^\beta} c_s(0) \, \dd s 
%	\Big| \leq \int_0^M \frac{Q_s(\beta)}{(s+1)^\beta} D(s,t) \, \dd s + 
%	C \int_M^\infty (s+1)^{-1-\beta} \, \dd s.
%	\]
%	The first term above converges to zero by dominated convergence, as 
%	$t \to \infty$. By taking $M \to \infty$ afterwards, the claim follows. 
\end{proof}

We have just proved Theorem~\ref{main thm} under the additional assumption that 
$\lim_{t \to \infty} Q_t(\beta)$ exists and is strictly positive. Since the definition
of $Q_t(\beta)$ involves the function $b_t(0) = \bsB[p_t](0)$ and thus the solution itself, such 
a condition is not desirable. The main step of our proof consists in showing that 
existence of $\lim_{t\to\infty} Q_t(\beta)$ already 
follows from the presence of a condensate, and that its positivity follows
from the finiteness of the condensate mass. Actually, we will show a bit more.
Note that weak convergence of $p_t$ to 
$\rho \delta_0 + q \, \dd x$ implies that 
\be \label{limit description of condensation}
\lim_{\eps \downarrow 0} \lim_{t \to \infty} \int_0^\eps p_t(x) \, \dd x = \rho.
\ee

\begin{proposition} \label{main part of proof}\ \\
	(i): Assume that (1)-(3) from Proposition \ref{b,c from B,C} hold, 
	and that the initial condition 
	$p_0$ is given by \eqref{initial condition}. Define $Q_t(\beta)$ as in 
	\eqref{Q_t}, with $\beta = \min \{\alpha, \alpha_0\}$. 	Then the following two statements are equivalent:
\begin{itemize}
	\item[$(a)$] $\lim_{t\to\infty} Q_t (\beta)$ exists and is finite. 
	\item[$(b)$] $(p_t)$ exhibits condensation, 
	i.e.\ \eqref{limit description of condensation} holds with $\rho >0$. \\
\end{itemize}
	(ii):
	If (a) holds with $\lim_{t \to \infty} Q_t(\beta) = 0$, and if there exist 
	$\delta > 0$ and $T > 0$ so that $p_t(x) \geq 0$ for all $x < \delta$ and 
	all $t > T$, 
	then we have $\rho = \infty$ in (b). \\[3mm]
	(iii): If (a) holds with $\lim_{t \to \infty} Q_t(\beta) > 0$, then 
	for 
	$0 < \eps < \delta$,
	\begin{equation}
	\label{explicit}
\lim_{t \to \infty} \int_{0}^\eps p_{t}(x) \, \dd x = \int_0^\infty \kappa(x) \, \dd x
+ \int_0^\eps  \frac{x^{\alpha-1} c_\infty(x)}{\gamma_\infty(x)} \,
\dd x,
\end{equation}
 where $\kappa(x)$ stands for 
	the right hand side of \eqref{condensate shape}. Moreover, 	
	\begin{equation}
		\label{condensate equality}
	\rho = \lim_{K \to \infty} \lim_{t \to \infty} \int_0^{K/t} p_t(x) \, 
	\dd x = \int_0^\infty \kappa(x) \, \dd x.
	\end{equation}	
\end{proposition}

\begin{proof}
	We start by proving parts (ii) and (iii), which also shows that $(a)$ 
	implies $(b)$ in part (i). For part (ii), note that the proof of Proposition 
	\ref{cond shape} actually shows that 
	\[
	\lim_{t \to \infty} \frac{Q_t(\beta)}{t} p_t(\tfrac{x}{t}) = 
	\e{-\gamma_\infty(0) x } 
	\Big( 1_{\{\beta = \alpha\}} x^\alpha \int_0^\infty 
	\frac{Q_s(\beta)}{(s+1)^\beta} c_s(0) \, \dd s +  1_{\{\beta = \alpha_0\}}
	x^{\alpha_0} \eta(0) \Big)
	\]
	under the condition that $(Q_s(\beta))$ is bounded. If $\lim_{t \to \infty} 
	Q_t (\beta)= 0$, this implies that $\lim_{t \to \infty} \frac{1}{t} p_t(x/t) = \infty$ 
	for all $x > 0$, and Fatou's lemma together with our positivity assumption 
	gives 
	\[
	\liminf_{t\to\infty} \int_0^{\delta/t} p_t(x) \, \dd x = 
	\liminf_{t\to\infty} \frac{1}{t} \int_0^\delta p_t(x/t) \, \dd x \geq 
	\int_0^\delta \liminf_{t \to \infty} \frac{1}{t} p_t(x/t) \, \dd x = 
	\infty.
	\]
	This shows (ii).  Now assume that $Q_\infty > 0$. 
	
	Then the rightmost equality in \eqref{condensate equality} is proved by observing
	that  
	\[
	\int_0^{K/t} p_t(x) \, \dd x = \frac{1}{t} \int_0^K p_t(x/t) \dd x 
	\to \int_0^K \kappa(x) \, \dd x
	\]
	as $t \to \infty$, by Proposition \ref{cond shape}, and then  
	taking $K \to \infty$. 
	Let us now write 
	\[
	\tilde p_t(x) := W_t \e{- t x \bar \gamma_{0,t}(x)} p_0(x) 
	\]
	for the second term in \eqref{rec}. Then for $t$ large enough, we have that
	\[
	\int_{K/t}^\eps \tilde p_t(x) \, \dd x = 
	(\tfrac{t+1}{t})^\beta t^{\beta - \alpha_0} Q_t(\beta)^{-1} 
	\int_K^{t \eps} \e{- x \bar \gamma_{0,t}(x/t)} x^{\alpha_0} \eta(x/t) 
	\, \dd x.
	\]
	Since $\bar \gamma_{0,t}(x/t)$ converges uniformly for $x \in [0,t\eps]$ 
	to $\gamma_\infty(x/t)$ (provided $\eps < \delta$), 
	and since $\gamma_\infty(x) > \gamma_- > 0$ 
	for $x < \delta$, we find that 
	\[
	\lim_{t \to \infty} \int_K^{t \eps} 
	\e{- x \bar \gamma_{0,t}(x/t)} x^{\alpha_0} \eta(x/t) \, \dd x \leq 
	Q_\infty^{-1}
	\int_K^\infty \e{- x \gamma_-/2} x^{\alpha_0} \sup_{z \leq \eps} |\eta(z)| \, \dd x \to 0
	\]
	as $K \to \infty$. We have shown 
	\begin{equation}
	\label{p tilde}	
	\lim_{t \to \infty} \int_0^\eps \tilde p_t(x) \, \dd x = \lim_{K \to \infty} 
	\lim_{t \to \infty} \int_0^{K/t} \tilde p_t(x) \, \dd x.
	\end{equation}
 	Now we turn to the first term of \eqref{rec}. We have 
	\begin{align}\nonumber
&\int _0^{t} \int_0^\eps   \frac{W_t}{W_s}   x^\alpha c_s(x)   
\e{- (t-s) x \bar \gamma_{s,t}(x)}\,\dd x\,\dd s\\
\begin{split}\label{eq1408-1}
&= \frac1{Q_t(\beta)} \int _0^{t/2}  \int_0^{(t-s) \eps}  \frac {(t+1)^{1+\beta}}{(t-s)^{1+\alpha}} \frac {Q_s(\beta) }{(s+1)^{1+\beta}} \, z^\alpha\, c_s(\tfrac z{t-s})   
\e{- z \bar \gamma_{s,t}(\tfrac z{t-s})}\,\dd z\,\dd s\\
&\quad + \int _0^{t/2} \int_0^\eps   \frac{Q_{t-s}(\beta)}{Q_{t}(\beta)} \frac{(t+1)^{1+\alpha}}{(t-s+1)^{1+\alpha}} \, x^\alpha c_{t-s}(x)   
\e{- s x \bar \gamma_{t-s,t}(x)}\,\dd x\,\dd s.
\end{split}
\end{align}
For the first integral in \eqref{eq1408-1}, 
we dominate the integrand uniformly by a constant multiple of
the integrable function
$$
\frac 1{(s+1)^{1+\beta}}\, z^\alpha\,e^{-\gamma_- z }.
$$
(Note that $\beta \leq \alpha$). If $\beta < \alpha$, then this term converges 
to zero. In the other case, we apply dominated convergence. In total, we obtain 
that the first term of \eqref{eq1408-1} converges to
$$
1_{\{\beta = \alpha\}} \frac 1{Q_\infty(\beta)} \int _0^{\infty}  \int_0^{\infty}   
\frac {Q_s(\beta)}{(s+1)^{1+\alpha}} \, z^\alpha\, c_s(0)   
\e{- z \gamma_{\infty}(0)}\,\dd z\,\dd s.
$$
Since we assumed $Q_t(\beta)$ to converge to a strictly positive limit, the 
integrand of the second integral in \eqref{eq1408-1}
is bounded by a constant multiple of the integrable function
$$
1_{\{x\leq \eps\}}\, x^\alpha e^{- \gamma_- sx}.
$$
It therefore converges, by dominated convergence,  to
$$\int_0^{\infty} 
	\int_0^{\eps} x^\alpha \, c_{\infty}(x)  \e{-s x \gamma_\infty(x)} \, \dd x \, \dd s
= \int_0^\eps  \frac{x^{\alpha-1} c_\infty(x)}{\gamma_\infty(x)} \,
\dd x.$$
If we put this together with \eqref{p tilde}, we have completed the 
proof of~\eqref{explicit}. The first claimed equality of 
\eqref{condensate equality} now follows by letting $\eps \downarrow 0$ in 
\eqref{eq1408-1}. \smallskip

Now we prove that $(b)$ implies $(a)$ in (i). 
Let $\mu_t(\eps) := \int_0^\eps p_t(x) \, \dd x$. Then, \eqref{rec} and 
	the definition of $Q_t(\beta)$ give 
	\begin{align}
		\nonumber
		Q_t(\beta) \mu_t(\eps) & = \int_0^t \dd s \int_0^\eps \dd x \, 
		(t+1)^{1+\beta} x^\alpha c_s(x) \frac{Q_s(\beta)}{(s+1)^{1+\beta}}
		\e{-(t-s) x \bar \gamma_{s,t}(x)} \\
		\nonumber
		& \quad + \int_0^\eps \dd x \, (t+1)^{1+\beta} 
		\e{-t x \bar \gamma_{0,t}(x)} x^{\alpha_0} \eta(x) \\
		\begin{split}
			\label{eqn main}
		&= \int_0^t \dd s \, Q_s \frac{(t+1)^{1+\beta}}{(s+1)^{1+\beta} 
		(t-s)^{1+\alpha}} \int_0^{\eps(t-s)} \dd z \, z^\alpha c_s(\tfrac{z}{t-s}) 
		\e{- z \bar \gamma_{s,t}(z/(t-s))} + \\
		& \quad + \int_0^{t\eps} \dd z \, \frac{(t+1)^{1+\beta}}{t^{1+\alpha_0}} 
		z^{\alpha_0} \e{-z \bar \gamma_{0,t}(z/t)} \eta(z/t). 
		\end{split}				
	\end{align}
	Let now $\eps < \delta$, where $\delta$ is as in Proposition \ref{b,c from B,C}. 
	Let us also assume that $t'$ is sufficiently large so that 
	$\bar \gamma_{0,t'}(x) \geq \gamma_-/2$ for all $0 \leq x \leq \delta$ and 
	all $t > t'$.
	This is possible since we assumed that $\bar \gamma_{0,t}$ 
	converges uniformly to $\gamma_\infty > \gamma_-$ on $[0,\delta]$. 
	Let us furthermore write $\bar \eta := \sup_{0 \leq x \leq \delta} \eta(x)$.
	Since $\beta \leq \alpha_0$, 
	the second line of \eqref{eqn main} is then bounded by 
	\begin{equation}
	\label{first part}
	2 \bar \eta
 	\int_0^\infty \dd z \, z^{\alpha_0} \e{- z \gamma_- / 2 } =: D_1 < \infty,
	\end{equation}
	for all $t$ sufficiently large. 
	
	For the first line of \eqref{eqn main}, we fix $q>1$, and for $t > 2 q$  
	we divide the integration range $[0,t]$ of the first integral 
	into $[0,q] \cup [q,t-q] \cup [t-q,t]$. We write 
	\[
	M_t := \sup_{s \leq t} Q_t(\beta), \qquad 
	\bar c := \sup \{ c_s(x): 0 \leq x \leq \delta, 0 \leq s < \infty\}. 
	\]
	Note that  $\beta \leq \alpha$ and thus
	$(t-s)^{-1-\alpha} \leq (t-s)^{-1-\beta}$ when $s \leq t-q$, $q>1$ and 
	$t> 2 q$. Thus the integral over $[0,q]$ is bounded by  
	\begin{equation}
	\label{second part}	
	M_q \, \bar c
	\int_0^q \dd s \, \frac{(t+1)^{1+\beta}}{(q+1)^{1+\beta}(t-q)^{1+\beta}}
	\int_0^\infty \dd z \, z^\alpha \e{- z \gamma_-/2} =:
	M_q C_1(q),
	\end{equation}
	for all $t$ sufficiently large. An elementary estimate shows that 
	\[
	\int_q^{t-q} \frac{(t+1)^{1+\beta}}{(s+1)^{1+\beta}(t-s) ^{1+\beta}} 
	\,\dd s \leq 2 \frac{(t+1)^{1+\beta}}{(t/2)^{1+\beta}} \int_q^{t/2} 
	s^{-(1+\beta)} \, \dd s \leq 2^{5 + 2 \beta} \beta^{-1} q^{-\beta},  
	\]
	and thus the integral from $q$ to $t-q$ is bounded by 
	\begin{equation}
		\label{third part}
	M_t \, \bar c \int_0^\infty \dd z \, 
	z^\alpha \e{- z \gamma_-/2} \int_q^{t-q} \dd s \, 
	\frac{(t+1)^{1+\beta}}{(s+1)^{1+\beta}(t-s)^{1+\beta}} =: 
	M_t \, C_2 q^{-\beta}
	\end{equation}
	for a suitable constant $C_2$ that is independent of $t$ and $q$. 
	Finally, for the integral from $q-t$ to $t$, we undo the change of 
	variable that transformed $x$ into $z = x/(t-s)$. Then this integral is 
	bounded by 
	\begin{equation}
	\label{fourth part}
	M_t \, \int_{t-q}^t \dd s \frac{(t+1)^{1+\beta}}{(s+1)^{1+\beta}} 
	\bar c  \int_0^\eps  x^\alpha \, \dd x\leq  
	 M_t \, q \frac{(t+1)^{1 + \beta}}{(t-q+1)^{1+\beta}} 
	 \bar c \int_0^\eps x^\alpha \dd x =: M_t \, C_3(q) \eps^{1 + \alpha},
	\end{equation}
	where $C_3(q)$ can be chosen independently of $t$.  
	We now use that we assumed condensation. Let $\rho>0$ be the mass of the 
	condensate, let $\rho_1 < \rho$, 
	and let $\eps_t$ be the smallest (unique, if $p_t > 0$) 
	solution of the equation 
	\[
	\int_0^{\eps_t} p_t(x) \, \dd x = \rho_1.
	\]
	Then $\mu_t(\eps_t) = \rho_1$ by definition, and 
	$\lim_{t \to \infty} \eps_t = 0$ by condensation. Now we choose $q = \bar q$ 
	large enough so that the right hand side of \eqref{third part} is less than
	$M_t \rho_1/3$. Then we choose $\bar t$ large enough so that for all 
	$t > \bar t$, we have 
	$C_3(\bar q) \eps_t^{1+\alpha} \leq \rho_1/3$. Thus for these $t$, 
	the right hand side of \eqref{fourth part} is less than $M_t \rho_1 /3$. 
	Writing $D_2 = M_{\bar q} C_1(\bar q)$ for the right hand side of 
	\eqref{second part}, an plugging all the parts back into \eqref{eqn main},
	we obtain the inequality 
	\[
		Q_t(\beta) \rho_1 \leq D_1 + D_2 + \tfrac{2}{3} M_t \rho_1,
	\]
	valid for all $t > \bar t$. Then, for all $t > \bar t$, we get 
	\[
	M_t \leq M_{\bar t} + \sup_{\bar t \leq s \leq t} Q_s(\beta) \leq 
	M_{\bar t} + \frac{D_1 + D_2}{\rho_1} + 
	\frac{2}{3} \sup_{\bar t \leq s \leq t} M_s = M_{\bar t} 
	+ \frac{D_1 + D_2}{\rho_1} + \frac{2}{3} M_t.
	\]
	We conclude that $M_t \leq 3 (M_{\bar t } + (D_1+D_2)/\rho_1)$ 
	and thus $t \mapsto Q_t(\beta)$ is bounded. 
	
	Next we show convergence of $(Q_t(\beta))$. We have 
\be 
\begin{split}\label{divide2}
	Q_t(\beta) \,\mu_t(\eps)&= 
	\int_0^{t/2}  \int_0^{(t-s)\eps}	 Q_s(\beta) 
	\big( \tfrac{t+1}{s+1} \big)^{1+\beta}\, 
	\tfrac{1}{(t-s)^{1+\alpha}} z^\alpha 
	c_s(\tfrac{z}{t-s}) \e{-z \bar \gamma_{s,t}(\frac{z}{t-s})} \, \dd z\,\dd s \\
	& \quad+ \int_{0}^{t/2} \int_0^{\eps} Q_{t-s}(\beta)
	\big( \tfrac{t+1}{t-s+1} \big)^{1+\alpha} x^\alpha 
	c_{t-s}(x) \e{-sx \bar \gamma_{t-s,t}(x)} \, \dd x\,\dd s \\
	& \quad + \int_0^{t\eps} \dd z \, \frac{(t+1)^{1+\beta}}{t^{1+\alpha_0}} 
		z^\alpha \e{-z \bar \gamma_{0,t}(z/t)} \eta(z/t). 
\end{split}
\ee 
For the first integral on the right hand side above, we dominate the integrand uniformly by a constant multiple of 
$$
\tfrac 1{(s+1)^{1+\alpha}}\, z^\alpha\,e^{-\gamma_- z }
$$
using that  $\bar \gamma_{s,t}(x) > \gamma_-$ for all $x\in[0,\eps]$ as long as  $t$ is sufficiently large.
Hence this integral converges to
$$
\int_0^\infty \int_0^\infty  \tfrac{Q_s(\beta)}{(s+1)^{1+\alpha}}\, 
	 z^\alpha 
	 c_s(0) \e{-z \gamma_{\infty}(0)} \, \dd z\,\dd s= \int_0^\infty \int_0^\infty W_s^{-1}\, 
	 z^\alpha 
	 c_s(0) \e{-z \gamma_{\infty}(0)} \, \dd z\,\dd s
$$
if $\alpha = \beta$, and to zero if $\alpha > \beta$. 
Further  the second term in \eqref{divide2} is for $t\geq t_0$  bounded by 
\[
U(\eps) =  C_1\,C_4\, 2^{1 + \alpha} \int_0^\infty   
\int_0^{\eps} x^\alpha  \e{-sx  \bar \gamma_{-}(x)} \, \dd x\,\dd s
\]
which converges to zero as $\eps\downarrow 0$.
By similar arguments, the third term converges to 
\[
\int_0^\infty z^\alpha \e{- z \gamma_{\infty}(0)} \eta(0) \, \dd z
\]
if $\beta = \alpha_0$, and to zero if $\beta < \alpha_0$. Defining 
\[
R := 1_{\{\beta = \alpha\}} \int_0^\infty \int_0^\infty W_s^{-1}\, 
	 z^\alpha 
	 c_s(0) \e{-z \gamma_{\infty}(0)} \, \dd z\,\dd s + 
	 1_{\{\beta = \alpha_0\}} \int_0^\infty z^\alpha \e{- z \gamma_{\infty}(0)} 
	 \eta(0) \, \dd z,
\]
we see from \eqref{divide2} and the assumption that $\lim_{\eps\downarrow0}\liminf_{t\to\infty} \mu_t(\eps)=\rho>0$ that 
\[
\limsup_{t\to\infty} Q_t(\beta) \leq \lim_{\eps\downarrow0} \frac {1}{\liminf_{t\to\infty} \mu_t(\eps)} \Bigl( R +U(\eps)\Bigr) = \frac{1}{\rho_0} R.
\]
Analogously, using that 
$\lim_{\eps\downarrow0}\limsup_{t\to\infty} \mu_t(\eps)=\rho>0$ we get that
\[
\liminf_{t\to\infty} Q_t(\beta) \geq \frac 1{\rho_0} R.
\]
Hence $(Q_t(\beta))$ converges to $R/\rho$. 	
\end{proof}

The proof of Theorem \ref{main thm} is thus finished. Under its assumptions, 
Proposition \ref{b,c from B,C} guarantees that the solution to 
\eqref{the main equation} fulfills equation \eqref{inhomogenous equation} with suitable 
properties of $b$, $c$ and $\gamma$. Since a finite, nonnegative condensation at 
$x=0$ is assumed, Proposition~\ref{main part of proof} guarantees that $Q_\infty$ 
exists, is finite and strictly positive. Then, Proposition \ref{cond shape} shows the 
claim of the proof. Note that equation \eqref{explicit} confirms that the shape 
of the bulk density needs to be 
\[ q(x) = 
\frac{x^{\alpha - 1} c_\infty(x)}{\gamma_\infty(x)} = 
\frac{x^\alpha \bsC[p_\infty](x)}{\bsB[p_\infty](x)},
\]
for $x>0$, a fact that already follows from the stationarity of $p_\infty$ for 
equation \eqref{the main equation}. Also, equation 
\eqref{condensate equality} shows that all of the condensate forms on a scale of 
$1/t$ as $t \to \infty$.

\section{Appendix: proof of proposition \ref{regularity from L1}}
	Fix $t > 0$. Then by our assumptions on $K_b$, 
	\[
	\int_0^\infty \sup_{x \in [0,\delta]} |\partial_x K_b(x,y)| \, 
	|p_t(y) - q(y)| \, \dd y < \infty.
	\]
	Thus Lebesgue's theorem allows us to differentiate under the integral sign 
	and obtain 
	\[
	\begin{split}
	\partial_x(\bsB[p_t](x) - \bsB[p_\infty](x)) & = \int_0^\infty \partial_x 
	K_b(x,y) (p_t(y) - q(y)) \, \dd y - \rho \partial_x K_b(x,0)\\
	& = \int_0^\infty \partial_x K_b(x,y) (p_t (\dd x) - p_\infty(\dd x)).
	\end{split}
	\]
	Furthermore, 
	\[
	\bsC[p_t](x) - \bsC[p_\infty](x) = \int_0^\infty K_c(x,y) 
	(p_t(\dd x) - p_\infty(\dd x),
	\]
	with a similar equation for $\bsB[p_t](x) - \bsB[p_\infty](x)$. Thus the 
	claim will be shown once we prove the following lemma:

\begin{lemma} \label{strongCond from L1 condition}
	Let $(p_t)$ and $p_\infty = \rho \delta_0 + q(x) \dd x$ 
	fulfill the assumptions of Proposition \ref{regularity from L1}. 
	Then we have 
	\[
	\lim_{t \to \infty} \sup_{0 \leq x \leq \delta} \Big| 
	\int_0^\infty h(x,y) \big(p_t(\dd y) -p_\infty(\dd y)\big)
	\Big| = 0
	\]
	for all $\delta > 0$ and all $h \in C_b([0,\delta] \times \bbR_0^+)$.
\end{lemma}

\begin{proof}
	First we rewrite 
	\begin{equation}
	\label{lemma eqn}	
	\begin{split}
	\int_0^\infty h(x,y) \big(p_t(y) - q(y)\big) \, \dd y - \rho h(x,0) = & 
	\int_0^\infty \big(h(x,y) - h(x,0)\big)\big(p_t(y) - q(y)\big) \, \dd y  \\ 
	& + h(x,0) \Big(
	\int_0^\infty (p_t(y) - q(y)) \, \dd y - \rho \Big)
	\end{split}
	\end{equation}
	By the weak convergence of $p_t \dd x$ to  
	$q \, \dd x + \rho \delta_0$, the 
	second line of \eqref{lemma eqn} converges to zero uniformly 
	in $x\in [0,\delta]$.
	For the first line, we fix $\eps > 0$.
	Since $h$ 
	is continuous, it is uniformly continuous on compact intervals. Therefore
	there exists $\tilde \delta > 0$ so that for all $x \in [0, \delta]$
	and all $y \in [0,\tilde \delta]$, we have $|h(x,y) - h(x,0)| < \eps$. Thus
	\begin{equation}
	\label{left interval}	
	\Big| \int_0^{\tilde\delta} \big(h(x,y)-h(x,0)\big) \big( p_t(y)-q(y)\big) 
	\, \dd y \Big| \leq \eps \int_0^{\tilde \delta} % \bar \delta
	 |p_t(y)| + |q(y)| \, \dd y.
	\end{equation}
	By decreasing $\tilde \delta$ if necessary and using condition
	\eqref{no stupid things near zero}, we find that 
	the latter integral is bounded  by a constant $C$, 
	and hence \eqref{left interval} is bounded by $C \eps$ uniformly in 
	$x \in [0, \delta]$.
	%, say) 	in $t$ by the assumption of $L^1$-boundedness. 
    Note that for sufficiently large $t$ this would follow already 
	from  weak convergence of $p_t$ if we had assumed that $p_t$ 
	is nonnegative. %Consequently,
	%thus \eqref{left interval} is bounded by $C \eps$ uniformly in 
	%$x \in [0, \delta]$. 
	For the remaining part of the first line of 
	\eqref{lemma eqn}, we estimate 
	\begin{equation}
	\label{right interval}	
	\int_{\tilde\delta}^\infty \big(h(x,y)-h(x,0)\big) \big( p_t(y)-q(y)\big) 
	\, \dd y \leq 2 \| h \|_\infty \int_{\tilde\delta}^\infty 
	|p_t(y)-q(y)| \, \dd y.
	\end{equation}
	By assumption \eqref{L1 away from zero}, 
	this integral tends to zero uniformly in $x$.  % Since
%	we have it, it follows that \eqref{right interval} converges to zero 
	%uniformly in $x$. 
	Putting \eqref{left interval} and \eqref{right interval}
	together, we finally get that %have thus shown that 
	\[
	\limsup_{t \to \infty} \sup_{0 \leq x \leq \delta} 
	\Big| \int_0^\infty \big(h(x,y) - h(x,0)\big)\big(p_t(y) - q(y)\big) 
	\, \dd y   \Big| \leq C \eps  
	\]
	and  the claim follows since $\eps>0$ was arbitrary.
	\end{proof}

\bigskip

\end{document}